\documentclass[11pt,oneside,reqno]{amsart}

 \usepackage[colorlinks=true,linkcolor=red,citecolor=blue]{hyperref}
 \usepackage{xcolor}
 \usepackage[final]{graphicx}
\usepackage{amsfonts}
\usepackage{pdfsync}
\usepackage{amsmath}
\usepackage{amssymb}
\usepackage{amsthm}
\usepackage{caption}
\usepackage{enumitem} 
\usepackage{dcolumn} 
\usepackage{mathrsfs}

\newtheorem{theorem}{Theorem}[section]
\newtheorem{lemma}[theorem]{Lemma}
\newtheorem{corollary}[theorem]{Corollary}
\newtheorem{proposition}[theorem]{Proposition}
\newtheorem{assumption}{Assumption}

\theoremstyle{remark}
\newtheorem{remark}{Remark}

\newcommand{\R}{{\mathord{\mathbb R}}}

\newcommand{\Z}{{\mathord{\mathbb Z}}}
\newcommand{\N}{{\mathord{\mathbb N}}}

\DeclareMathOperator{\supp}{supp}
\DeclareMathOperator{\Int}{Int}
\DeclareMathOperator{\adh}{Adh}
\newcommand{\ext}{{\rm ext}}
\newcommand{\cA}{\mathcal{A}}
\newcommand{\cB}{\mathcal{B}}
\newcommand{\cD}{\mathcal{D}}

\newcommand{\cL}{\mathcal{L}}

\newcommand{\di}{\,\mathrm{d}}

\newcommand{\norm}[1]{\left\|#1\right\|}
\newcommand{\bvec}[1]{\boldsymbol{#1}}
\newcommand{\inprod}[2]{\left \langle #1 , #2\right \rangle}

\newcommand{\Dom}{\mathcal{D}}

\newcommand{\Ran}{\operatorname{Ran}}

\newcommand{\Tr}{\operatorname{Tr}}

\title[Spectrum of linearized Vlasov--Poisson]{Spectrum of the linearized Vlasov--Poisson equation around steady states from galactic dynamics}

\date{\today}

\author[M. Moreno]{Matías Moreno}
\address{Departamento de Ingenier\'ia Matem\'atica and Centro
de Modelamiento Matem\'atico (CNRS IRL 2807), Universidad de Chile, Beauchef 851, Santiago, Chile}
\email{mmoreno@dim.uchile.cl}
\author[P. Rioseco]{Paola Rioseco}
\email{paola.rioseco@uchile.cl}
\author[H. Van Den Bosch]{Hanne Van Den Bosch}
\email{hvdbosch@dim.uchile.cl}

\begin{document}

\maketitle

\begin{abstract}
  We study the linearized Vlasov--Poisson equation in the gravitational case around steady states that are decreasing  and continuous functions of the energy. We identify the absolutely continuous spectrum and give criteria for the existence of oscillating modes and estimate their number. Our method allows us to take into account an attractive external potential. 
\end{abstract}

\medskip

\noindent{\bf Acknowledgments.}  
The authors received support from the Center for Mathematical Modeling (Universidad de Chile \& CNRS IRL 2807) through ANID/Basal projects \#FB210005. P.R acknowledges ANID/Fondecyt project \#322--0767. H.VDB. and M.M. received funding from ANID/Fondecyt project \#1122--0194.

\section{Introduction}
In this paper, we study the linearization of the gravitational Vlasov--Poisson system around steady states that are monotone functions of the microscopic energy. 
Our goal is to obtain precise results on the spectrum of the linearized operator, depending on the classical trajectories in the potential generated by the steady state solutions. 
We obtain two results of different natures. 

First, for the case of spherically symmetric perturbations in dimension $3$ and sufficiently regular steady states, we show that the self-interaction is a relatively trace-class perturbation. This statement has important consequences for the dynamics of solutions to the linearized equations. It means that the time evolution when restricted to this spectral subspace is closely resembles the time evolution without the interaction.

\medskip

The second part of our work considers the point spectrum of the linearized operator, building on previous work by Mathur \cite{1990SM} and the recent articles of Had\v{z}i\'{c} et al. \cite{hadzic2023,hadvzic2022}. We provide a condition on the steady state to determine whether bound states exist or not, see Theorem~\ref{thm:eigenvalues-intro}.  Bound states correspond to oscillating modes in the linearized dynamics of the Vlasov--Poisson system. These results complement those in the recent preprint \cite{hadzic2023}, which appeared during the preparation of this manuscript. In \cite{hadzic2023}, the linearized equation is analyzed around steady states where the angular momentum takes a fixed value $L_0 >0$. In this case, there is no matter in a neighborhood of the origin. In this paper, we study the complementary case, in the sense that the steady states are functions of the energy only. 
If $\phi$ does not depend on angular momentum, the support of the steady states is an open set in the six-dimensional phase-space, instead of the five-dimensional manifold of trajectories with fixed value of $L =|\bvec{x} \times \bvec{v}|$. At the technical level, this makes the treatment more involved and also explains the difference in the conclusions.

\medskip
The Vlasov-Poisson system  in $d=3$ is the system
\begin{align} \label{eq:VP-intro-1}
    \partial_t f &= -\bvec v\cdot\nabla_{\bvec x} f + \nabla_{\bvec x}(U_f+U_{\rm ext})\cdot\nabla_{\bvec v} f,  \\
 \label{eq:VP-intro-2}
    U_f(t,\bvec x)&=-(1/|\cdot|\ast \rho_f)(t,\bvec x), \quad \rho_f(x) := \int_{\R^2} f(x,v) \di v 
\end{align}
We allow for radially symmetric external potentials $U_{\rm ext} \in C^2(\R^3)$. Global existence for this system has been shown in \cite{pfaffelVP}.
This equation describes the evolution of a density $f$, depending on the position $x\in \R^3$ and momentum (or velocity) $v \in \R^3$, subject to its own Coulomb potential $U_f$. It can be thought of as a mean-field description of a large number of particles interacting through Newtonian gravity. 

The gravitational Vlasov--Poisson system has infinitely many steady state solutions.
A direct way to obtain steady states is to define the microscopic energy $E= E(x,v) = v^2/2 + U(x)$ and look for a solution of the form $f(x,v) = \phi(E)$, for a suitable function $\phi$. Such an ansatz gives rise to a stationary solution if there is a potential $U$ satisfying 
\[
U(x)-U_{\rm \ext}(x) = \int |x-y |^{-1} \int \phi(v^2/2 + U(y))\di v \di y.
\]
Equivalently, by defining
\[
\Phi(u) :=\int_{\R^3} \phi(\bvec v^2/2 + u)\di \bvec v = \int_{\R} \phi(E)\sqrt{2(E-u)_+} \di E, \quad \text{ and }\rho_{\ext}= -\Delta U_{\ext}
\]
and applying $-\Delta$ to both sides, $U$ solves 
\[
-\Delta( U ) = 4 \pi \Phi(U) - 4 \pi \rho_\ext.
\]
When $\phi$ is of the form $(E_0-E)_+^n$ and $\rho_\ext = 0$, this is the classical Lane-Emden equation (since then $\Phi(u)= c_n (E_0-u)_+^{n+3/2}$). For each choice of $R_0>0$, it has a unique, radial, negative solution in $B(0,R_0)$ that vanishes at $|\bvec x|= R_0$. For more general choices of $\phi$ and nonzero right-hand-sides, one may for instance \cite[Theorem 11.5]{GilbargTrudinger} to obtain solutions for all H\"{o}lder continuous $\Phi$ and $\rho_\ext$.

A more systematic way to show the existence of steady states is to look at minimizers of the so-called Energy-Casimir functionals. We refer to \cite{Rein98,GuRe01} for a detailed and careful treatment, but the method allows to show existence of minimizers that are compactly supported, of finite mass and with finite density, and again of the form $f=\phi \circ E$. A generalization of the method allows to include a dependence on angular momentum $L := |x \wedge v|$ as well, such that the steady states take the form $f(x,v) = \widetilde \phi (E(x,v),L(x,v))$ for suitable functions $\widetilde \phi $ depending on two real variables.

Among this giant collection of steady states, the so-called polytropes play a special role in the physics literature. The \emph{isotropic polytropes} are steady states depending on energy only, of the form $\phi(e)= C(E_0-e)_+^n$, where $C$ is a constant and $E_0$ a parameter, and the \emph{anisotropic} generalization $\widetilde \phi(e,l) = C l^{-s} (E_0 -e)^n$. Here $e$ and $l$ are just labels for the variables. Their names merely stem from the analogy with equations of state for fluids, and all these solutions are radially symmetric.

\medskip

In this paper, we will consider a general steady state that is a decreasing and sufficiently regular function of the energy only, see Assumption~\ref{hypothesis}. Once a steady state is identified, one may wonder about its stability, i.e., about the long term behaviour of small perturbations of the steady state.

Orbital or neutral stability, i.e., initial conditions close to a steady state remain close to this steady state, has been established for the case $U_{\rm ext} = 0$. This stability can be obtained in the variational context, see e.g. \cite{wolansky1999nonlinear,guo1999variational,guo2000generalized,dolbeault2004asymptotic, sanchez2006orbital,lemou2008orbital}, under the assumption that minimizers are isolated, which has been proven for the polytropic solutions in \cite{schaeffer2004steady}. In  \cite{lemou2011new, lemou2012orbital} the authors overcome the need for isolated minimizers using decreasing rearrangements on the level sets of the ground state energy.

The question is then whether initial conditions close to a steady state converge or \emph{relax} to some steady state, or might oscillate forever. In the plasma context (repulsive interactions), the former occurs and the phenomenon is called \emph{Landau Damping}. It has been shown to hold in the spatially homogeneous context by \cite{cMcV11} (see also \cite{bedrossian2016landau} for further developments).
It turns out that the classical Penrose criterion \cite{penrose1960} for the absence of exponentially growing modes at the linear level also ensures nonlinear damping for (Gevrey) regular perturbations.

The analogous question for  spatially inhomogeneous steady states remains widely open. Recent progress in this direction includes \cite{faou2021} where damping in the linearized equation associated to a simplified model is studied, and \cite{pausader2022}, where decay is shown for a repulsive kinetic gas orbiting a point charge.

\medskip

In the context of astrophysical steady states, nonlinear damping seems out of reach at the moment. However, the behaviour of the linearized equation plays a crucial role in the proofs of Landau damping, and this fact motivates the present study. The linearized equation has a long history in the physics literature. 
In \cite{doremus1973stability, gillon1976stability, kandrup1985simple}, the linear stability of the linearized equation is studied, while
 Mathur investigates the possibility of eigenvalues or oscillating modes in \cite{1990SM}, which assumes a confining background potential. 
His treatment has been placed on a rigorous footing in \cite{hadvzic2022} or the monograph \cite{kunze2021}. 
\medskip

\subsection{The Linearized Vlasov--Poisson System}

In this paper, we study the linearization of \eqref{eq:VP-intro-1}
around a fixed, radially symmetric stationary solution $f_0$, such that
\begin{equation} \label{eq:VP-intro-stationary }
   \bvec v\cdot\nabla_{\bvec x} f_0 - \nabla_{\bvec x}(U_{f_0}+U_{\rm ext})\cdot\nabla_{\bvec v} f_0 = 0 .
  \end{equation}
To shorten the notation, we define
$$
U = U_{f_0}+U_{\rm ext}.
$$

We refer to the introduction in \cite{hadvzic2022} for details on the derivation of the linearized equation and its interpretation, and limit ourselves to a formal calculation for a small and regular perturbation of the stationary state in the form $f_0 + g$.
Dropping quadratic terms in $g$ from equation~\eqref{eq:VP-intro-1} gives
\begin{equation}
    \label{eq:lin_VP}
\partial_t g = -\cL g + \nabla_{\bvec v} f_0 \cdot \nabla_{\bvec x} U_g,
\end{equation}
with the potential 
$$
U_g(\bvec x):=-\int_{\R^d \times \R^d} \frac{g(\bvec y,\bvec v)}{|\bvec x-\bvec y|}\di \bvec v \di \bvec y,
$$
and the Liouville operator associated to the stationary state
$$
\cL := \bvec v\cdot \nabla_{\bvec{x}} - \nabla_{\bvec{x}} U\cdot \nabla_{\bvec{v}}.
$$
The term $\cL$ determines the transport of $g$ in phase space along the trajectories associated to the potential $U$, while the second term shows how the perturbation interacts gravitationally with the steady state $f_0$.

Following Antonov \cite{1987Antonov} (see again \cite{hadvzic2022}) for a detailed derivation), this equation can be rewritten as a system for the projections $g_\pm$ of $g$ on the subspaces of functions that are even ($g_+$) and odd ($g_-$) in the velocity variable. The key observation is that $\int_v g_- = 0 $, and thus the odd part does not contribute to the interaction term. 
Also, $\cL$ and $\nabla_{\bvec{v}}$ map even functions in $v$ to odd functions in $v$ and vice versa, so the equation becomes the system
\begin{align*}
       \partial_t g_+ &= -\cL g_- \\
    \partial_t g_- &= -\cL g_+ + \nabla_{\bvec v} f_0 \cdot \nabla_{\bvec x} U_{g_+}.
\end{align*}
By taking a time-derivative of the second equation, we obtain a closed equation for the odd part
\begin{equation}
    \label{eq:lin_VP_squared}
\partial_t^2 g_- = \cL^2 g_-  +  \cB g_-:= - \cA g_-,
\end{equation}
with 
\begin{align*}
    \cB g (\bvec x,\bvec v) := \nabla_{\bvec v}
    f_0(\bvec x,\bvec v)\cdot \nabla_{\bvec x}\int_{\R^d\times \R^d} \frac{\cL   g(\bvec{y},\bvec{u}) }{|\bvec x-\bvec y|} \di \bvec y \di \bvec u .\\
\end{align*}
The operator $\cA$ is called Antonov operator, and its spectrum determines the behaviour of the linearized equation.
It was introduced by Antonov \cite{Antonov61}.
The key point is that it can be realized as a self-adjoint operator and its spectrum belongs to $[0, +\infty)$, so there are no exponentially growing modes. Once this is settled, one can ask more detailed questions about the spectral properties of $\cA$, that determine the linearized dynamics. 
For instance, if $\cA$ has positive eigenvalues, there are perturbations that do not decay at the linearized level, but rather oscillate forever. To the contrary, perturbations in the continuous spectral subspace exhibit mixing in a time-averaged sense, and those in the absolutely continuous subspace exhibit mixing as $t \to \infty$, by the RAGE Theorem, see for instance \cite{Teschl}[Section 5.2] for a nice introduction.

As noted by Mathur, in analogy with the case of Schrödinger operators, the \emph{main} part of the operator is the free transport part $-\cL^2$. We will recall below how its spectrum can be found, and it turns out to be purely absolutely continuous (except for its kernel), which implies (weak) decay or \emph{mixing} for the time evolution. 
The question is then how the interaction term $\cB$, which is a relatively compact perturbation, changes the spectrum. Again, from the analogy with a Schrödinger operator, one could expect eigenvalues to appear in the \emph{gaps} of the essential spectrum. A powerful and versatile tool to obtain information on these is the Birman-Schwinger principle, which reduces an eigenvalue for unbounded operator to a question for related compact operators. In the current context, this approach was followed by \cite{1990SM, kunze2021, hadvzic2022, hadzic2023}  and recently extended to the Einstein-Vlasov system in \cite{gunther2022}. For an overview of these results, we refer to the recent review \cite{rein2023}.

While $\cA$ has the same essential spectrum as $-\cL^2$, it is not clear (at least to us) that it preserves the absolutely continuous \emph{quality} of the essential spectrum and the ensuing mixing property. There is a huge literature on Schr\"{o}dinger operators where conditions on potentials ensuring the presence or absence of singularly continuous spectrum are given, and it would be interesting to see if some of the techniques can be adapted to deal with perturbation of $-\cL^2$, and, maybe counter-intuitively, describe \emph{mathematical scattering} (in the sense of e.g., \cite{Yafaev92}) for a bounded physical system.

\medskip

From now on, we will consider a fixed external potential $U_{\rm ext}$ and a fixed steady state solution $f_0$ of \eqref{eq:VP-intro-stationary }. We assume throughout that $f_0$ and $U_{\rm ext}$ satisfy the following assumptions, which include for instance the \emph{sufficiently regular} isotropic polytropes. 

\begin{assumption}
\label{hypothesis} ~ 
\begin{enumerate}[itemsep=0.1em]
\item\label{hyp:on_U_ext}The external potential $U_{\rm ext}$ is negative, nondecreasing, of class $C^2$. Furthermore,  $U_{\rm ext}$ is radially symmetric, subharmonic and tends to $0$ at infinity. We will abuse notation and write $U_{\rm ext}(\bvec x) = U_{\rm ext}(|\bvec x|)$.
\item \label{hyp:isotropic} The steady state $f_0$ is rotation invariant and depends on the energy $E(\bvec x, \bvec v)= U(\bvec{x}) + 
|\bvec v|^2/2$ only. We define $f_0= \phi \circ E$. The function $\phi$ has compact support $[U(0), E_0]$ for some $E_0 <0$.
    \item \label{hyp:on_phi'} $\phi$ is differentiable in the interior of its support, $\phi'(E) < 0$ and $\phi'(E(\bvec x, \bvec v))$ is integrable.
\end{enumerate}
\end{assumption}

We now group some basic consequences of this Assumption.
Point \ref{hyp:on_U_ext}
implies that for all $e<0$, the set $\{(\bvec{x}, \bvec{v})| E(\bvec{x}, \bvec{v}) \le e\}$ is bounded. 
Hence, since $f_0$ is supported (Point \ref{hyp:isotropic}) in $\{(\bvec{x}, \bvec{v})| E(\bvec{x}, \bvec{v}) \le E_0\}$, it has compact support in $\R^3\times \R^3$. For later reference, we define
\begin{equation} \label{eq:def_Omega_0}
    \Omega_0 = \Int \supp (f_0) \subset \R^3 \times \R^3 \quad 
    \text{and } \quad B(0,R_0) = \Int \supp (\rho_{f_0})  \subset \R^d.
\end{equation}

Point \ref{hyp:on_phi'} implies that $\phi$ is continuous inside its support and hence
the mass density $\int f_0(\bvec x, \bvec v) \di \bvec v$ is finite.
Also note that the total potential $U$ is subharmonic as the sum of $U_{\rm \ext}$ (Assumed in Point~\ref{hyp:on_U_ext}) and $U_{f_0}$ (from the Poisson equation). In radial coordinates, this means that $-\partial_r (r U'(r)) \le 0$.

\begin{remark}
    Hypotheses \ref{hyp:isotropic} and \ref{hyp:on_phi'} are automatically satisfied for a large range of steady states. They are satisfied by the polytropic solutions $\phi(E) = (E_0-E)^k_+$ for all $k \in (0, 7/2)$.
\end{remark}
\begin{proof}
    The only property that has to be checked is that $\phi' \circ E \in L^1(\Omega_0)$.  Since $E$ depends on $v:=|\bvec v|$ and $r:= |\bvec{x}|$ only, we can first change to radial coordinates and then, from $(r,v)$ to $(r,E)$, to obtain
\begin{equation} \label{eq:identity-int}
    \int_{\Omega_0} \phi'(E) = (4\pi)^2 \int_{U(0)}^{E_0} \phi'(E) \int_{0}^{U^{-1}(E)} r^{2}[2(E-U(r))]_+^{1/2}\di r \di E.
\end{equation}
Since $\phi'(E)= - k [E_0-E]_+^{k-1}$, the integrand is bounded  except for the divergence (if $k<1$) when $E$ approaches $E_0$.
The innermost integral is bounded away from zero for $E$ close to $E_0$, so integrability of $\phi'(E)$ in $\Omega_0$ is equivalent to integrability of $\phi'$ in $(U(0), E_0)$.
For the polytropic solutions $\phi(e) = (E_0 - e)_+^{k}$, this holds if and only if $k \ge 0$. 
\end{proof}



\subsection{Classical mechanics and the period function.}

We can now turn our attention to the transport operator $\cL$
 and we state some basic properties from classical mechanics that are important for its behaviour.
The trajectories of a particle in the radially symmetric potential $U$ are confined to the plane perpendicular to the (conserved) angular momentum $\bvec L:= \bvec x \wedge \bvec v$, assuming it is nonzero. It is then convenient to take polar coordinates $(r, \varphi)$ in this plane, oriented such that an increase in $\varphi$ corresponds to a positive rotation in the direction $\bvec{L}$. In these variables, Newton's equations of motion reduce to
\begin{align*}
   \begin{cases}
       \dot{\varphi} &= L/r^2 \\
       \ddot{r} &= - \partial_r V_{\rm eff} (r)
   \end{cases}, \quad \text{ where } L := |\bvec{L}| ,
   \text{ and } V_{\rm eff} (r) := U(r) + L^2 / (2 r^2).
\end{align*}
The second equation describes the evolution of the radial variable as a one-dimensional system with Hamiltonian $v_r^2/2 + V_{\rm eff}(r)$. Once it is solved, the angle $\varphi$ can be found by integrating the first equation (Kepler's Area Law).
Since we assumed that $U(r)$ is bounded, we find that, for any nonzero $L$, the effective potential $V_{\rm eff}(r)$ diverges to $+\infty$ as $r$ approaches zero: particles with nonzero angular momentum do not reach the origin. Also, both terms in $V_{\rm eff}(r)$ tend to zero as $r$ tends to infinity. Equilibria of the dynamical system correspond to critical points of $V_{\rm eff}(r)$. 
We compute
\[
V_{\rm eff}'(r) = U'(r) - L^2/r^3 = r^{-1} \left(r U'(r) - L^2/r^2\right).
\]
The subharmonicity of $U$ from Assumption~\ref{hyp:on_U_ext} and the definition of $U_{f_0}$, implies that $r U'(r)$ is nondecreasing, while $L^2/r^2$ is decreasing from $+\infty$ to $0$, there is exactly one point $r_*= r_*(L)$ where $V_{\rm eff}'(r) = 0$, and it is a minimum. For each value of the energy $E$ in $(V_{\rm eff}(r_*), 0)$, there are two turning points $r_- (E,L)< r_+(E,L)$, the solutions to $U(r)+L^2 /(2r^{2})= E$.
For these values of $E$, there exists a bounded trajectory, where $r$ oscillates between $r_-$ and $r_+$, and the angle $\varphi$ increases according to Keppler's Second Law.
The time it takes to perform one oscillation can be obtained from the classical formula

\begin{equation} \label{eq:intro_frequencies}
  T(E,L) =2\int_{r_-(E,L)}^{r_+(E,L)} \frac{\di r}{\sqrt{2(E-V_{\rm eff}(r)) }} :=  \frac{2\pi}{\omega_r(E,L)}.
\end{equation}
Here, $\omega_r(E,L)$ is the angular frequency of the radial oscillations with energy $E$ and angular momentum $L$\footnote{there is one such orbit for each choice of the direction of $\bvec L$}, and it will play an important role throughout this paper. 
For generic potentials $U$, $\omega_r (E,L)$ is different from frequency with which $\varphi$ increases by $2 \pi$, and thus the trajectories in the physical space are not ovals, but rather rosette-like figures that may or may not close after several radial periods.

When $E$ approaches its minimum value $E_{\min}(L):=V_{\rm eff}(r_*)$,
$r_-$ and $r_+$ approach $r_{*}$, so the range of integration goes to zero, but also the integrand diverges. 
This trajectory corresponds to a circular orbit with $r(t)=r_*$. In this \emph{small oscillations}-limit, $\omega_r(E,L)$ has a continuous extension and approaches 
\begin{equation}\label{eq:circular-frequencies}
\omega_r(E_{\min}(L), L) := \sqrt{V_{\rm eff}''(r_*)} 
= \sqrt{U''(r) + 3 U'(r)/r}.
\end{equation}

In the remainder of the paper, it will be convenient to also consider $L_{\max}(E)$, which is the inverse function of $L \mapsto E_{\min }(L) $ and represents the maximum value of angular momentum for which bounded trajectories exist at the energy $E$.

It remains to consider the case $L=0$. In this case, $\bvec x$ and $\bvec v$ are co-linear for all times and the particle moves on a line through the origin. The equations of motion describe one-dimensional oscillations in the potential $U(|x|)$. 
We slightly abuse notation and define $r_-(E,L=0)=0$ and $r_+(E, L= 0)= U^{-1}(E)$ for each $E \in [U(0), 0)$. Then, 
the turning points of the linear trajectory are $\pm r_+(E,0)$. 
We \emph{define}
\begin{equation} 
  T(E,L=0) =2\int_{0}^{r_+(E,0)} \frac{\di r}{\sqrt{2(E-U(r)) }} :=  \frac{2\pi}{\omega_r(E,L=0)},
\end{equation}
such that $T$ and $\omega_r$ (and also $r_\pm$) are continuous functions of $E$ and $L$, but $T(E,L=0)$ is only \emph{half} the period of the line-trajectory oscillating between $\pm r_+(E,0)$.

The situation is illustrated in Figure~\ref{fig:plano_E_L}. Through each point of the phase space with $E(\bvec x, \bvec v) < 0$ passes a unique bounded trajectory, whose period depends only on $E(\bvec x, \bvec v)$ and $L(\bvec x, \bvec v)$. 
\begin{figure}
    \centering
    \includegraphics{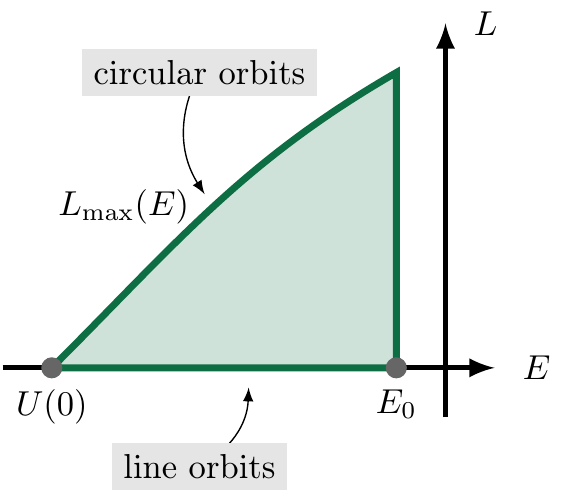}
    \caption{Sketch of the values taken by $E,L$ in the support of the steady state. }
    \label{fig:plano_E_L}
\end{figure}

\medskip
In action-angle coordinates, the Liouville operator $\cL$ takes the convenient form
$$
\cL = \omega_r(E,L) \partial_{\theta_r} + \omega_\varphi (E,L) \partial_{\theta_\varphi}
$$

The frequencies $\omega_r(E,L)$ and $\omega_\varphi (E,L)$ determine the properties of the Liouville operator $\cL$. In the remainder of the paper we will be mostly interested in radially symmetric perturbations. Functions in this subspace are independent of $\theta_\varphi$ and we are concerned with $\omega_r$ only. Heuristically, if $\omega _r$ varies fast with $E,L$, different orbits evolve at different speeds and phase-space mixing occurs. Hence the frequencies determine the \emph{free} evolution generated by $-\cL^2$, that does not take into account the self-gravitation of the perturbation. In this paper, we show that the frequencies are also crucial when studying the \emph{interacting}
operator $-\cL^2 - \cB$.

\subsection{Main Results}
We  summarize our main results in the following two theorems. We denote by $\cA_{\rm rad}$, $\cL_{\rm rad}$ the restriction of the respective operators to the subspace of rotationally invariant functions.

\begin{theorem} \label{thm:trace_class_intro}
Under Assumption~\ref{hypothesis},  in the suitable Hilbert space,   $\cB_{\rm rad}$ is relatively trace class with respect to $-\cL^2_{\rm rad}$. Hence $\cA_{\rm rad}$ has the same absolutely continuous spectrum as the non-interacting operator $-\cL^2_{\rm rad}$.
  \end{theorem}
For a more precise statement, see Section~\ref{sec:ac_spec}. The essence is that for \emph{many} radial perturbations, phase space mixing occurs. A complete identification of the spectral quality, however, would require to show that $\cA|_{\rm rad}$ has no singularly continuous spectrum, which is beyond the scope of this paper.
In the previous references \cite{hadvzic2022,kunze2021}, the relative compactness of $\cB$ was used to identify the essential spectrum of $\cA$. We recall this fact in Proposition~\ref{prop:Antonov} below. While relative compactness also holds without restricting to radial perturbations, the relative trace-class property only seems to hold in the radial subspace. 

In order to state our results on the point spectrum of $\cA$, we need to define
\begin{equation} \label{eq:def_omega_*}
    \omega_* = \min_{\Omega_0} \omega_r (E,L).
\end{equation}
and the function
\begin{equation}
    \rho_*(\bvec x) := \int_{\R^d} \frac{\phi'(E(\bvec x, \bvec v)) \omega_r^2(E(\bvec x, \bvec v),L(\bvec x, \bvec v))}{\omega_r^2(E(\bvec x, \bvec v),L(\bvec x, \bvec v)) - \omega_*^2} \di \bvec v.
\end{equation}
It is known (\cite{hadvzic2022}, or see Section~\ref{sec:preliminaries}  below) that $\sigma_{\rm ess}(\tilde \cA_{\rm rad})$ has a gap $(0,\omega_*^2)$.

\begin{theorem} \label{thm:eigenvalues-intro}
Under Assumption~\ref{hypothesis} 
\begin{enumerate}
    \item \label{it:intro-bound}
If $|\bvec x|^{-1} \rho_*(\bvec x) $ belongs to $L^1(B(0,R_0))$ and
    $$
    \norm{|\bvec x|^{-1} \rho_*(\bvec x) }_{L^1}\le m \in \N
    $$ 
    then $\cA_{\rm rad}$ has at most $m-1$ (possibly $0$) eigenvalues below $\omega_*^2$
        \item  \label{it:intro-divergence} If $\rho_*$ does not belong to $L^1$, and $\omega_*$ is not attained at the circular orbits, then $\cA_{\rm rad}$ has at least one eigenvalue below $\omega_*^2$, and the corresponding eigenfunction is odd in $v$. 
\end{enumerate}
\end{theorem}

This theorem gives a criterion, depending on the regularity of $\phi'$ and how $\omega_r$ approaches its minimum value, to distinguish two different scenarios. In the \emph{regular} case, the number of eigenvalues can be controlled and no eigenvalues occur when $\phi'$ is  \emph{small}. In the \emph{divergent} case, eigenvalues occur even for arbitrarily small $\phi'$. Whether they are finite in number or not, remains an open question.

In principle, this theorem gives no information when $|\bvec x|^{-1}\rho_*$ is not integrable, but $ \omega_*$ is attained at the circular orbits or $\rho_*$ is integrable.
However, we are not aware of realistic scenarios that belong to this case.
\medskip

For the circular orbits it can be shown from the expression in equation~\eqref{eq:circular-frequencies} that $\omega_r(E,L_{\max}(E))$ is decreasing in $E$.
 To check the hypothesis that $\omega_*$ is not attained at the circular orbits, one could for instance check that $\omega_r(E_0, L_{\max}(E_0)) > \omega_r(E_0 ,0)$.

On the other hand, if the minimum frequency is attained for some orbit inside the support of the steady state, then the denominator in the definition of $\rho_*$ diverges quadratically and the resulting function is never in $L^1$.  
Finally, for \emph{many} polytropic solutions, $\omega_r(E,L)$ is expected to be strictly increasing in $E$ and decreasing in $L$, such that $\omega_* = \omega(E_0,0)$. In this case, the factor $(\omega_r(E,L) - \omega_* )^{-1}$ diverges as
\begin{align}
    \left(\partial_E \omega_r(E_0,0)(E-E_0) + \frac{\partial_L^2\omega_r(E_0,0)}{2} L^2 \right)^{-1}.
\end{align}
A computation with this ansatz gives the following corollary.
\begin{corollary} \label{cor:polytropes}
Consider a steady state where $\phi(t) = [t-t_0]_+^n$ and $n \in (0, 7/2)$. Furthermore, assume that $\omega_* = \omega_r(E_0, 0)$ and that $\partial_E \omega_r(E_0,0) < 0$, and $\partial_L^2 \omega_r(E_0,0) > 0$. Then \ref{it:intro-bound} of Theorem~\ref{thm:eigenvalues-intro} holds.
\end{corollary}

At first sight, this result might seem at odds with the one from \cite{hadzic2023} where the Birman-Schwinger operator is unbounded for exponents $n<1$. However, in the case studied there, the steady state is of a distribution of the form $\widetilde\phi (E,L) =(E_0-E)_+^n \delta(L-L_0)$, and as such, more singular than the case under consideration here.

\subsection{Outline of the paper}
In Section~\ref{sec:preliminaries}, we define the operators, study their basic properties and essential spectrum, and fix some notation that will be used throughout the paper. Then, we specify to the radially symmetric subspace. In Section~\ref{sec:ac_spec}, we give the short proof of Theorem~\ref{thm:trace_class_intro}. The main technical part of the paper is Section~\ref{sec:EVs} where the Birman-Schwinger operator is introduced in order to prove Theorem~\ref{thm:eigenvalues-intro}. The final Section~\ref{sec:polytropes} contains the estimates needed for Corollary~\ref{cor:polytropes}.

\section{Preliminaries} \label{sec:preliminaries}

In this section, we look at the basic properties of $\cA$, defined as a differential expression in \eqref{eq:lin_VP_squared}. As a first step, we need to realize it as a self-adjoint operator on a domain in a suitable Hilbert space. 

 For our purposes, it is sufficient to realize $-\cL^2$ as the Friedrichs' extension of $\langle \cL f, \cL g \rangle_{L^2(\Omega_0)}$ initially defined on functions in $C^\infty(\Omega_0)$. For the identification of the domain of $\cL$, see \cite{hadvzic2022}.

Next, we rewrite the interaction term $\cB$ from \eqref{eq:lin_VP_squared} in a convenient form.
We use that
$$
\nabla_{\bvec v} f_0(\bvec x,\bvec v) = \bvec v \phi'(E(\bvec x,\bvec v)).
$$
Hence,
$$
\cB f(\bvec x,\bvec v) = \phi'(E(\bvec x,\bvec v))(\bvec v \cdot \bvec \nabla_{\bvec x}) \int \frac{1}{|\bvec x - \bvec y|}\int \cL f(\bvec y , \bvec u) \di \bvec  u \di \bvec y .
$$
After integration by parts, several applications of Fubini's Theorem and the commutation of $\cL$
with the function $\phi'(E)$, we can write the quadratic form in a convenient form
$$
\langle g, \cB f \rangle =  \cD \left( \int  |\phi'(E(\cdot, \bvec v))| \cL g(\cdot, \bvec v) \di \bvec v, \int   \cL f(\cdot, \bvec v) \di \bvec v \right),
$$
where the Coulomb-inner product is defined by
\begin{equation} \label{eq:def_Coulomb}
    \cD(\rho_1, \rho_2) := \int_{\R^d} \int_{\R^d} \frac{\bar \rho_1(\bvec{x}) \rho_2(\bvec{x})}{|\bvec{x} - \bvec{y}|} \di \bvec{x} \di \bvec{y}.
\end{equation}

In order to work with a manifestly symmetric expression, we define the multiplication operator
$$
A f= |\phi'|^{-1/2}(E) f
$$
and to simplify notation, 
define $\mu$ from $L^2(\Omega_0)$ to $ L^2(B(0,R_0))$ by
  \begin{equation} \label{eq:def_mu}
        \mu (h)(\bvec{x}) = \int |\phi'|^{1/2}(E(\bvec x, \bvec v)) h(\bvec x, \bvec v) \di \bvec v.
  \end{equation}
With this notation, we obtain
$$
\langle g, A \cB A^{-1}f \rangle =  \cD \left( \mu(\cL g), \mu(\cL f) \right).
$$
This motivates the definition
\begin{equation} \label{eq:def_operator}
    \tilde \cA : = A \cA A^{-1}= -\cL(1 - K_\phi) \cL \quad \text{ with } K_\phi f  = |\phi'|^{1/2} U_{|\phi'|^{1/2} f}.
\end{equation}
Note that, if an (odd) function $g$ solves \eqref{eq:lin_VP_squared}, then
$\tilde g := A g$ solves the equation
$$
\partial_t^2 \tilde g =  - \tilde \cA \tilde g.
$$
In the remainder of the paper, we study the spectral properties of $\tilde \cA$ as an unbounded operator in $L^2(\Omega_0)$.
We first show that $\tilde A$ is a self-adjoint operator, as a relatively compact perturbation of $-\cL^2$.

\begin{proposition}\label{prop:Antonov}
    The operator $K_\phi$ as defined above, is Hilbert-Schmidt, and we have Antonov's bound $\norm{K_\phi |_{\Ran (\cL) }} \le 1$.
\end{proposition}
\begin{proof}
 From the Cauchy-Schwarz inequality, we see that the operator $\mu$ is bounded from $L^2(\Omega_0)$ to $L^2(B(0,R_0))$ by
 $$
\sup_{\bvec x \in B(0,R_0)} \int |\phi'(E((\bvec x, \bvec v))| \di \bvec v,
 $$
 provided that this expression is finite. 
 We compute
\begin{align*}
    \int |\phi'(E((\bvec x, \bvec v))| \di \bvec v &= -4 \pi \int_0^{\infty} \phi'(v^2/2 + U(\bvec x)) v^{2} \di v \\
    &= - 4\pi \int_{U(\bvec x)}^{E_0} \phi'(E) \left(2(E - U(\bvec x))\right)^{1/2} \di E \\
    &\le - 4 \pi \int_{U(0)}^{E_0} \phi'(E) \left(2(E - U( 0))\right)^{1/2} \di E ,
   \end{align*}
 which is indeed finite by \eqref{eq:identity-int}.
 
 Now $K_\phi = \mu^* R \mu$, where $R$ is convolution with $1/|\bvec x|$ restricted to $B(0,R_0)$. The kernel belongs to $L^2(B(0,R_0) \times B(0,R_0))$, hence the operator is Hilbert-Schmidt. 

For Antonov's bound, we use separate arguments to study radially symmetric functions in the range of $\cL$, and non-radial functions without any orthogonality condition. 

{\noindent \bf Radially symmetric functions.} The argument for radially symmetric functions can be found in details in \cite[Chapter 2]{kunze2021} or the Appendix of \cite{lemou2011new}, while the method goes back to \cite{kandrup1985simple, perez1996stability}.
Since these authors do not include an external potential, we sketch the steps to check that the proof carries over to this case without modifications.
First, we define the auxiliary multiplication operator $T:= (x\cdot v)$.
A calculation shows that
$$
[\cL, T] = |v|^2 - |x| U'(x), \qquad [\cL,[\cL, T]] = - T(4\pi (\rho_{0} + \rho_{\rm ext}) + U'(x)/|x|).
$$
The key identity is (see the Appendix of \cite{lemou2011new} for the delicate justification of the integration by parts)
\begin{align*}
    \norm{\cL u}^2 
    &= \norm{T\cL T^{-1}u}^2  \\
   & \qquad + \int_{\R^3}(4\pi (\rho_{0}(\bvec x)+ \rho_{\rm ext}(\bvec x)) + U'(\bvec x)/|\bvec x|)\int_{\R^3}|u|^2(\bvec x, \bvec v) \di \bvec v \di \bvec x,
\end{align*}
where we have defined $\rho_{\rm ext} = \Delta U_{\rm ext}/(4\pi) \ge 0$.
The term that we have to control is 
$$
\inprod{\cL u}{K_\phi \cL u} = D(|\phi'|^{1/2} \cL u, |\phi'|^{1/2} \cL u) = \frac{1}{4 \pi}\norm{\nabla U_{|\phi'|^{1/2} \cL u}}^2.
$$
Using Newton's theorem, the definition of $\cL$, and the divergence theorem twice, and the definition $v_r = \bvec v \cdot \hat{ \bvec x}$ simplifies the integrand to 
\begin{align*}
    \nabla U_{|\phi'|^{1/2} \cL u}(\bvec x)
     &= 4 \pi \hat{\bvec  x}\int_{\R^3} v_r  (|\phi'|^{1/2}u)(\bvec x, \bvec v) \di \bvec v.
\end{align*}
The Cauchy-Schwarz inequality gives
\begin{align*}
    |\nabla U_{|\phi'|^{1/2} \cL u}|^2(\bvec x)& \le (4\pi)^2 \int_{\R^3} v_r^2 |\phi'|(\bvec x, \bvec v) \di \bvec v \int_{\R^3} |u|^2(\bvec x,\bvec v)) \di \bvec v,
\end{align*}
and one last application of the divergence theorem yields
\begin{align*}
 \int_{\R^3} v_r^2 |\phi'|(E(x,v)) \di v 
 &= -  \int_{\R^3}v_r (\hat x \cdot\nabla_v) \phi(E(x,v)) \di v\\
 & = \int_{\R^3} (\hat x \cdot\nabla_v) v_r \phi(E(x,v)) \di v\\
 &= \rho_0(x).
\end{align*}
Hence, we obtain
\begin{align*}
   \inprod{\cL u}{K_\phi \cL u} \le 4\pi \int_{\R^3}  \rho_0(x) \int_{\R^3} |u|^2(x,v) \di v < \norm{\cL u}^2 . 
\end{align*}
Since $K_\phi$ is a compact operator, this implies the bound $\norm{K^{\rm rad}_\phi|_{\Ran(\cL)}} < 1$. 

{\noindent \bf Nonradial functions.}
Assume that $\lambda \neq 0$ is an eigenvalue of $K_\phi$ with associated eigenfunction $g$. The eivenvalue equation implies that $g(\bvec x,\bvec v) = |\phi'(E(\bvec x,\bvec v))|^{1/2} \chi(\bvec x)$ for some function $\chi$ depending on the space-variable only. The equation for $\chi$ is
 \begin{align*}
     \lambda \chi(\bvec x) 
     &= \int_{\R^3} \frac{\chi(\bvec y)}{|\bvec x - \bvec y|} \int_{\R^3}|\phi'(E(\bvec y,\bvec v))| \di \bvec v \di \bvec y 
     &=\int_{B(0,R_0)} \frac{\chi(\bvec y)}{|\bvec x - \bvec y|} \rho_{|\phi'|}(\bvec y) \di \bvec y ,
 \end{align*}
If $g$ is orthogonal to functions invariant under rotations of the full phase-space, then $\chi$ is orthogonal to radial functions. After projecting on spherical harmonics, we may assume that $\chi(\bvec x) = j(r) Y_l^m(\theta, \phi)$, where $Y_l^m$ are the standard spherical harmonics, $\bvec x$ has spherical coordinates $(r,\theta, \phi)$ and the orthogonality to radial functions translates as $l\neq 0$. This also implies that $\chi$ extends to an $L^2$-function in all of $\R^3$ (since in particular, $\int_{B(0,R_0)} \chi \rho_{|\phi'|} = 0 $).   Applying $-\Delta$ to both sides of the equation shows that $\chi$ is an eigenfunction of a Schrödinger operator
 \begin{equation*} 
      \left( -\Delta - \frac{4\pi}{\lambda} \rho_{|\phi'|} \right)\chi =  0,
 \end{equation*}
 or equivalently
  \begin{equation} \label{eq:for_j}
      j''(r)  + \frac{2}{r}   j'(r) = V_{l, \lambda}(r) j(r) \quad \text{ with } V_{l,\lambda}(r) := \frac{l(l+1)}{r^2} - \frac{4\pi}{\lambda} \rho_{|\phi'|}(r).
 \end{equation}
The Vlasov--Poisson equation provides us with a comparison function. The equation for the potential reads
 \begin{align*}
      U_{f_0}''(r) + \frac{2}{r} U_{f_0}'(r) =  4 \pi \rho_{f_0}(r) = 4\pi \int_{\R^3}\phi(E(r,\bvec v)) \di \bvec v. 
 \end{align*}
 Taking one more radial derivative gives the following equation for $h(r):= U_{f_0}'(r)$
  \begin{align*}
     h''(r) + \frac{2}{r} h'(r)  -\frac{2}{r^2}h(r) & = 4 \pi \int_{\R^3}\phi'(E(r,\bvec v)) (U'_{f_0}(r) +U'_{\rm ext }(r)) \di \bvec v \\
     &=- 4\pi \rho_{|\phi'|}(r)(h(r) +U'_{\rm ext }(r))
 \end{align*}
or
  \begin{equation} \label{eq:for_h}
      h''(r)  +\frac{2}{r} h'(r) = V_{1,1} h(r) -4 \pi \rho_{|\phi'|(r)}U'_{\rm ext }(r). 
 \end{equation}
 In the translation invariant case $U'_{\rm ext} \equiv 0$, $h$ is a solution of \eqref{eq:for_j} in the case $\lambda=l=1$ and this is just an implication of the fact that translations of the stationary state remain stationary. In any case, a standard comparison of solutions to Sturm-Liouville problems \eqref{eq:for_j} and \eqref{eq:for_h} shows that $\lambda  \le 1$ and that the equality case only occurs when $j=h$. To see this, define $W = h' j - j' h $ and compute
 $$
 (r^2 W(r))' = \bigl(V_{1,1}(r) - V_{l,\lambda}(r)\bigr) h(r) j(r)  -4 \pi \rho_{|\phi'|(r)}U'_{\rm ext }(r) j(r).
 $$
 Assume that $\lambda >1$, such that $V_{1,1} < V_{l,\lambda}$.  Since $h$ is positive, this shows that $(r^2 W)'$ has the sign of $- j$.
 If $j$ has no zeroes, we use that $r^2 W$ vanishes at $0$ and $+\infty$ and arrive at a contradiction. Else, define $r_*$ to be the first zero of $j$ and assume that $j>0$ on $(0, r_*)$. Then $W(r_*)= - j'(r_*) h(r_*) < 0$ and we arrive, again, at a contradiction.
 \end{proof}

 \begin{corollary}
     The sum $-\cL^2 + \cL K_\phi \cL$ can be realized as a self-adjoint operator, which is relatively compact with respect to $-\cL^2$. We will denote this operator by $\widetilde A$ throughout the paper.
 \end{corollary}
 \begin{proof}
The quadratic form 
 $\inprod{\cL g}{K_\phi \cL g}$ is relatively bounded with respect to  $\inprod{\cL g}{\cL g}$ with relative bound $\lambda < 1$, except (possibly, in the case $U_\ext = 0$) for a $3$-dimensional subspace. Therefore, the quadratic form $\inprod{\cL g}{\cL g} - \inprod{\cL g}{K_\phi \cL g}$ is closed on $\Dom (\cL)$, positive semi-definite and $\widetilde A $ can be realized as a Friedrich's extension~\cite[Theorem X.23]{ReedSimonIII}.
 \end{proof}

\section{Essential spectrum}\label{sec:ac_spec}

The essential spectrum of $\tilde \cA$ coincides with that of $-\cL^2$ ( $\cL$ is skew-adjoint). 
In order to study this spectrum, it is convenient to write the Liouville operator in action-angle coordinates associated with the trajectories in the central potential $U:= U_{\rho_0} + U_{\rm ext}$.
As discussed in the introduction, to each pair $(E,L)$ with $L<L_{\max}(E)$ corresponds to periodic trajectories in the plane $(r, v_r)$, between the turning points $r_{\pm}(E,L)$. 
For each fixed $L$, the \emph{radial phase space} $(r,v_r)$ can be parametrized in \textbf{Action-Angle} coordinates, where the action corresponds to the energy and
The angle $\theta_r$ parameterizes this trajectory such that $\dot \theta_r = \omega_r(E,L)$ and we fix it such that $\theta_r(E,L, r_+(E,L)) = 0$. 
Since a change of coordinates can be realized as a unitary transformation,
$$
\sigma (i\cL^{\rm rad} ) = \sigma (i\omega_r(E,L) \partial_{\theta_r})=\{k \omega_r(E,L)| k \in \Z, (E,L) \in \Omega_0\}.
$$
Since $\omega_*>0$, this spectrum has a gap around zero. Away from zero (the contribution of $k=0$), it is a union of bands. If the set of critical points of $\omega_r(E,L)$ in $\Omega_0$ has Lebesgue measure zero, these bands are absolutely continuous.

In dimension $d=3$, the expression for $K_\phi^{\rm rad}$ can be simplified by using Newton's theorem \cite[Sec. 9.7]{liebloss2001}, in the form
$$
U_\rho(\bvec x) =-4 \pi \int_0^{+\infty}  \frac{\rho(r)}{\max(|\bvec x|,r)} r^2 \di r,
$$
valid for radial functions $\rho$.

\begin{proposition} \label{prop:rel-trace}
In $d=3$, $K_\phi^{\rm rad}$ is trace class. Hence $\tilde \cA$ is a relatively trace class perturbation of $-\cL^2$.
\end{proposition}

This proposition implies that the absolutely continuous spectrum of $\cA_{\rm rad}$ is the same as the absolutely continuous spectrum of $-\cL^2$, by the Kato-Rosenblum Theorem \cite{kato1957perturbation, rosenblum1957perturbation} and \cite[Theorem XI.9]{ReedSimonIII}, and the associated subspaces are unitary equivalent. Within the absolutely continuous spectral subspace, the dynamics of the system corresponds to mixing, by the Riemann-Lebesgue Lemma, see e.g. \cite[Section 5]{Yafaev2004trace} for the second-order case.

\begin{proof}
Radially symmetric functions in phase space depend on the three variables $r = |\bvec x|$, $v= |\bvec v|$ and $w:= \bvec x \cdot \bvec v $, having as range the set
$$
T := \{(r,v,w) | r \ge 0, |w| \le r v\}.
$$
Integration over $\bvec v$ reduces to integration over the last two variables with weight $2\pi v/r$, and integration over the total phase space corresponds to integration over $T$ with weight $(2\pi)^2 v r$.
By Newton's theorem, 
we can represent $K_\phi$ as an operator in $L^2(T, 2\pi v r\di r \di v \di w)$ with kernel
$$
\kappa_\phi (r,v, w ; r', v', w') := \frac{|\phi'(E(r,v))|^{1/2}|\phi'(E(r',v'))|^{1/2}}{\max(r, r')}.
$$
Hence,
\begin{align*}
   \Tr (K_\phi) = \int \kappa_\phi 
   &= \int_{T}\frac {|\phi'|(E(r,v))}{r} (2\pi)^2 v r \di r \di v \di w  \\
   &= 8\pi^2 \int_0^\infty \int_0^\infty \left(-\partial_v \phi(E(r,v)) \right) v r \di v \di r,
\end{align*}
where the integration over $w$ contributes $2 r v$.
Integration by parts (without boundary terms since $\phi(E_0)= 0$), gives
\begin{align*}
   \Tr (K_\phi) = 
   & 8\pi^2 \int_0^\infty \int_0^\infty r\phi(E(r,v)) \di v  \di r,
\end{align*}
which is finite as the integral of a continuous function over the bounded domain $\{(r,v)| E(r,v) < E_0 \} $.
\end{proof}



\section{Eigenvalues} \label{sec:EVs}
In this section, we study eigenvalues in the gap in the essential spectrum of $\cA_{\rm rad}$, the interval $(0,\omega_*^2)$. Throughout this section we will drop the subscript $r$ in $\omega_r$. 

Now assume that $\lambda$ is an eigenvalue of $\cA_{\rm rad}$ with associated eigenfunction $g$. Then a short computation shows that $\sqrt{-\cL^2 -  \lambda} \, g$ is an eigenfunction of  
$$
K(\lambda) := \sqrt{-\cL^2 - \lambda}^{-1} \cL_{\rm rad} K_\phi \cL_{\rm rad}  \sqrt{-\cL^2 - \lambda}^{-1},
$$
with eigenvalue $1$. 
We first state a direct consequence of the Birman-Schwinger principle that is the first point of Theorem~\ref{thm:eigenvalues-intro}. 

\begin{proposition}
\label{prop:bound-trace}
If 
$$
\frac{\phi'(E)}{|x|} \frac{1}{\omega(E,L) - \omega_*} \in L^1(\Omega_0)
$$
we have the bound
$$
\norm{K(\lambda)} \le \Tr(K(\lambda)) \le \int_{\Omega_0} \frac{|\phi'(E)|}{|\bvec x|} \frac{\omega^2(E,L)}{\omega^2(E,L) - \omega_*^2} \di \bvec v \di \bvec x
$$
\end{proposition}

\begin{proof}
 We write $\inprod{g}{K(\lambda) g} = D(\mu_\lambda (g), \mu_\lambda (g))$, where 
 
\begin{equation} \label{eq:def_mu_lambda}
 \mu_\lambda = \mu \cL(-\cL^2 - \lambda)^{-1/2}   
\end{equation}

with $\mu$ defined in \eqref{eq:def_mu}.  
For $k \in \Z $, the function  $k \mapsto |k \omega|/\sqrt{k^2\omega^2 - \lambda}$ attains its maximum  at $k= \pm 1$, hence
$$
|\cL|(-\cL^2 - \lambda)^{-1/2} \le \omega(E,L) (\omega(E,L)^2 - \lambda)^{1/2}.
$$
By using Newton's theorem as in the proof of Theorem~\ref{prop:rel-trace}, we obtain the result.
\end{proof}
\begin{figure}
    \centering
   \includegraphics[width= 0.48\textwidth]{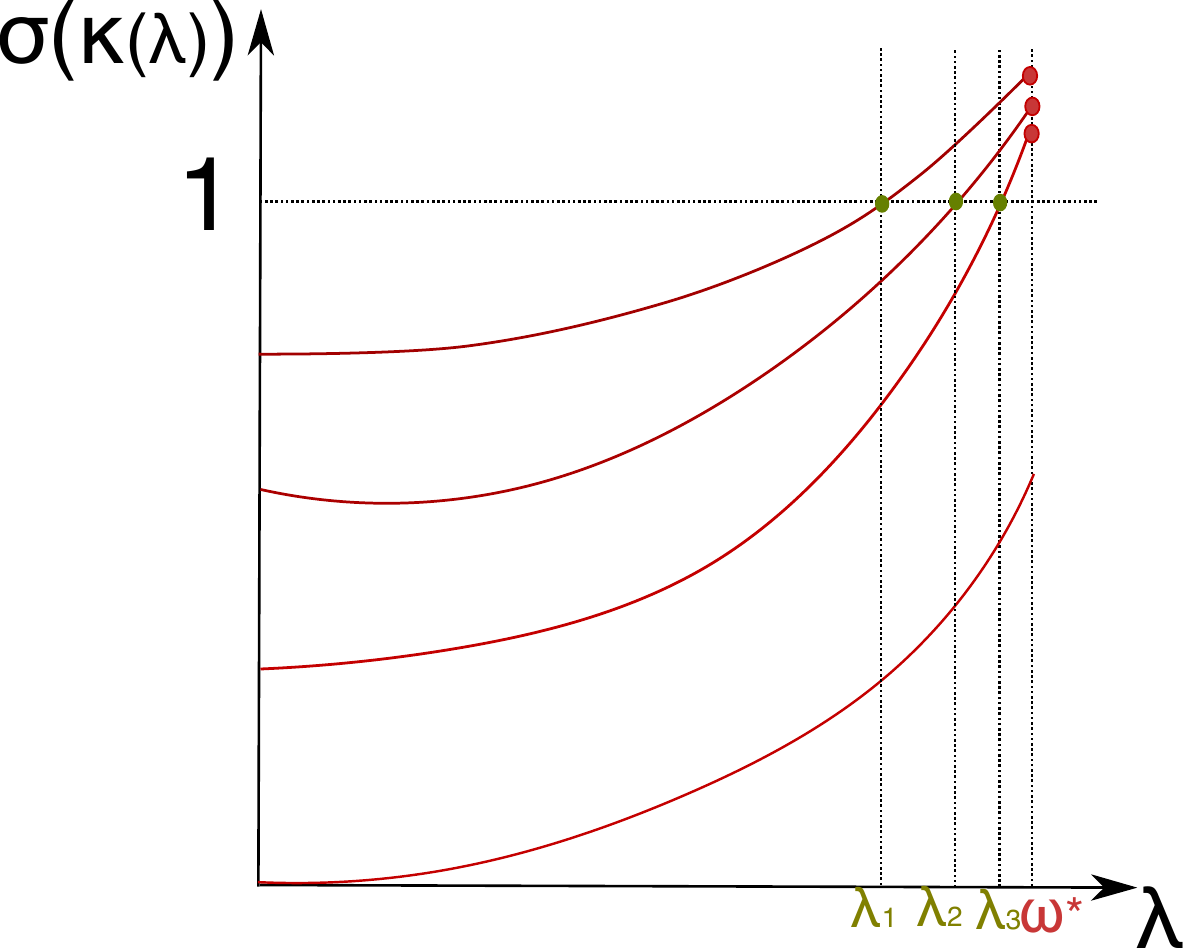}
    \includegraphics[width= 0.48\textwidth]{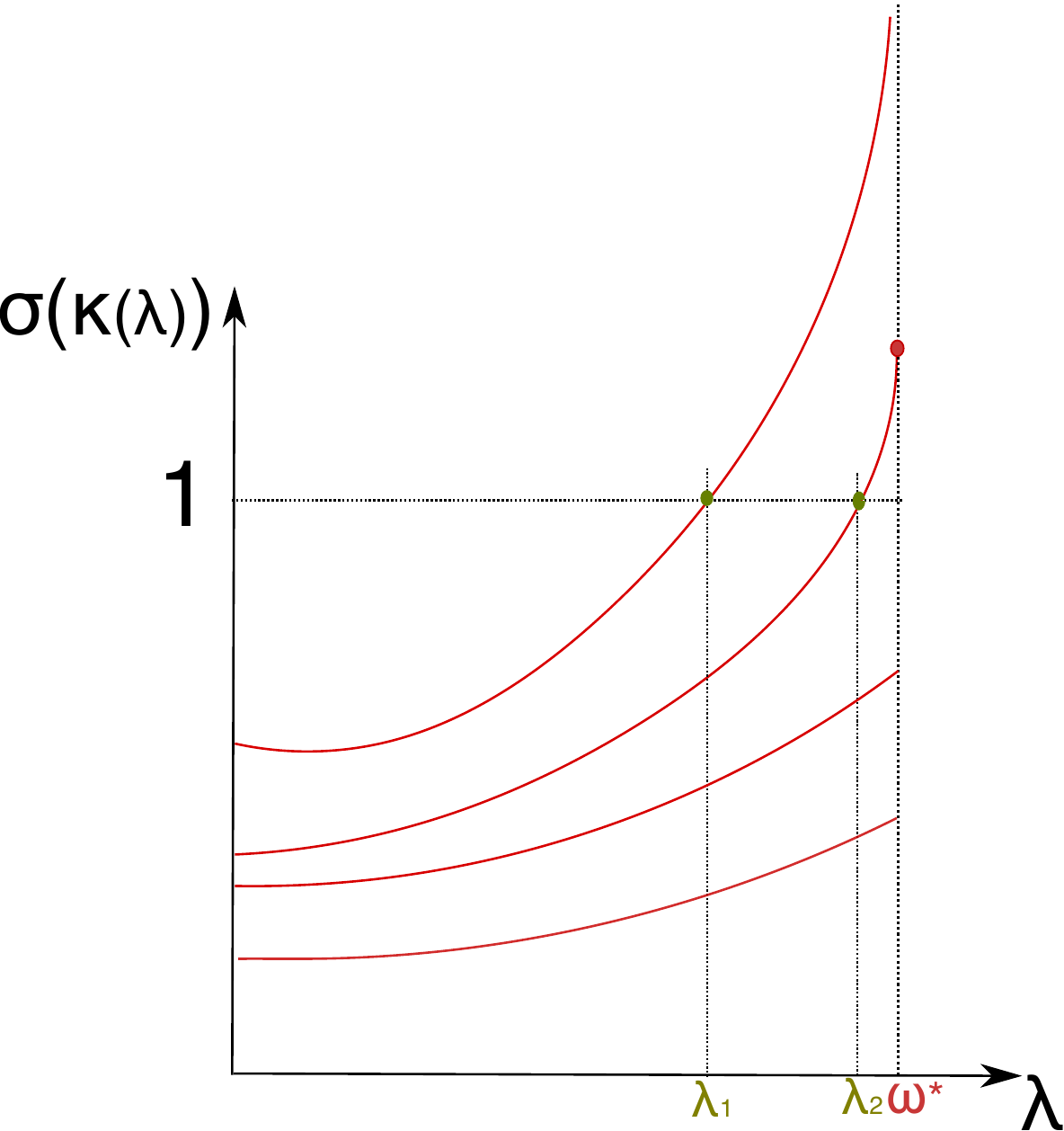}
    \caption{Sketch of possible scenarios for the eigenvalues of the Birman-Schwinger operator. Left : if $\rho_*(x)/|x|$ is integrable, all eigenvalues are bounded. Right : if $\rho_*$ is not integrable, the first eigenvalue diverges as $\lambda$ approaches $\omega_*$.  }
    \label{fig:Birman_Schwinger}
\end{figure}
\begin{proof}[Proof of Theorem~\ref{thm:eigenvalues-intro}.\ref{it:intro-bound}]
The eigenvalues of $K (\lambda)$ are continuous and increasing in $\lambda \in [0, \omega_*)$, and coincide with those of $K_\phi$ at $\lambda = 0$. A qualitative picture is given in Figure~\ref{fig:Birman_Schwinger}. For a given value of $\lambda$, the number of eigenvalues of $\tilde \cA_{\rm rad}$ below $\lambda$ equals the number of eigenvalues of $K(\lambda)$ above $1$. By Proposition~\ref{prop:bound-trace}, the trace of $K(\lambda)$ gives an upper bound on this number given by $\norm{|\bvec x|^{-1} \rho_*(\bvec x) }_{L^1}$. Therefore, if this quantity is less than $m\in\N$, then there must be at most $m-1$ eigenvalues of $\mathcal{A}_{\text{rad}}$ strictly below $\omega_{\ast}$.
\end{proof}

Recall that $\Omega_0$ is the support of the steady state in phase space. In order to state the next proposition, we also need to define
\begin{equation} \label{eq:def_Omega_eps}
    \Omega_\epsilon = \Omega_0 \cap \{ (\bvec x, \bvec v) | r_+(E,L) > r_-(E,L) + \epsilon \}.
\end{equation}
Hence, $\Omega_\epsilon$ is a subset of the support bounded away from the circular orbits.

\begin{proposition} \label{prop:divergence}
If $\phi'(E)/(\omega(E,L)- \omega_*)$ is not integrable in $\Omega_\epsilon$ with some $\epsilon > 0$, then $\norm{K(\lambda)|_{\rm odd}}$ diverges as $\lambda \nearrow \omega^2_*$. 
\end{proposition}

The relevance of the restriction to $\Omega_\epsilon$, is that it allows to construct test functions, localized near a point where $\omega$ approaches $\omega_*$, which give rise to a nonzero density. Since we will need subsets of phase space determined by fixed values of $(E,L)$, we introduce the following notation:
$$
\hat{A} := \{(E(\bvec x,\bvec v), L(\bvec x, \bvec v)) | (\bvec x, \bvec v) \in A\}
$$

\begin{lemma} \label{lemma_technical}
Under the Hypotheses of Proposition~\ref{prop:divergence}, there exists $\hat\Lambda\subset \hat \Omega_\epsilon$ such that  $\Lambda := (E,L)^{-1}(\hat\Lambda)$ has the following properties.
\begin{enumerate}
    \item \label{it:divergence_point}$\phi'(E)/(\omega(E,L)- \omega_*) |_{\Lambda}$ is not integrable
    \item  \label{it:suff small} There exists $r_1$ such that  for all $(E,L) \in \hat \Lambda$ and for all $r \in (r_1, r_+(E,L))$ we have $\cos(\theta_r(E,L, r)) >0$  
    \item \label{it:bounded orbits}There exists  $r_2 > r_1$ such that for all $(E,L) \in \Lambda$, and all $\theta_r \in (0,\pi/3)$, we have $r(E,L, \theta_r) \ge r_2$
\end{enumerate}
\end{lemma}
This situation is sketched in Figure~\ref{fig:onion-orbits}.
We first prove the Proposition and then the Lemma.

\begin{figure}
    \centering \includegraphics[]{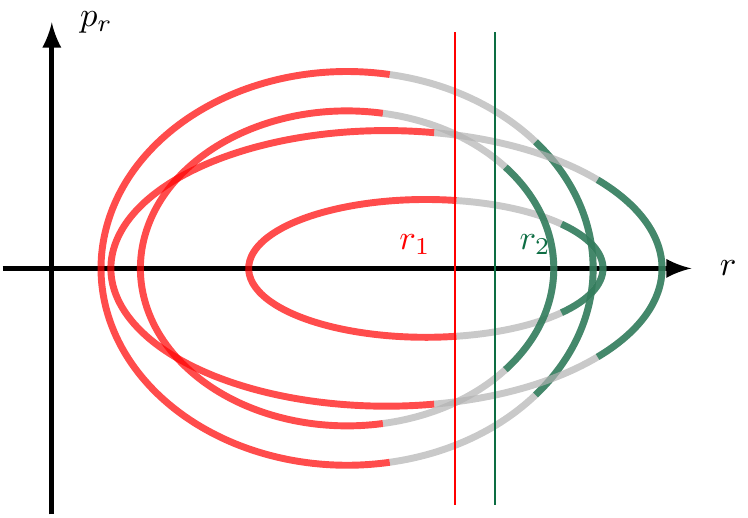}
    \caption{Sketch of some orbits in phase space for nearby values of $E$ and $L$. The red part of each orbit ($\cos(\theta_r)<0$) lies to the right of $r_1$ and the green part ($\cos(\theta_r)> 1/2$) lies to the left of $r_2$. }
    \label{fig:onion-orbits}
\end{figure}

\begin{proof}[Proof of Proposition~\ref{prop:divergence}]
\noindent \textbf{Coordinates for the odd subspace}
Rotationally symmetric functions $f$ can be written as functions of $r:=|\bvec x|$, $v:= |\bvec v|$ and $w:= \bvec x \cdot \bvec v$. Such an $f$ is odd in $\bvec v$, if and only if it is odd in $w$. In order to study the norm of $K (\lambda)$, it is more convenient to use the action variables $E,L$, and associated angle $\theta_r$, defined such that $\theta_r=0$ corresponds to the outer turning point $r_+$. In these coordinates functions $\tilde f(E,L, \theta_r)$ that are odd in $\theta_r$ correspond to odd functions in $\bvec v$. This is because a change of sign in $\theta_r$ corresponds to a change of sign in the radial component of $w$. Hence, a function in the odd subspace can be written as a function of the variables $(E,L, \theta_r)$ that is odd in the last variable. Such a function can be decomposed in turn in a Fourier series
$$
g = \tilde g(E,L, \theta_r) = \sum_{k=1}^{\infty} \tilde g_k(E,L) \sin(k \theta_r).
$$
The advantage of this form is that 
$$
\cL g = \sum_{k=1}^{\infty} \tilde g_k(E,L) k \omega(E,L) \cos(k \theta_r).
$$

\noindent\textbf{Definition of the test function.} 
In order to show divergence of the operator norm, it is sufficient to construct, for each $\lambda$, a test function $h_\lambda$ in the odd subspace, such that
\begin{equation}
    \label{eq:Rayleigh-diverges}
    \lim_{\lambda \nearrow \omega_*^2}\frac{\langle h_\lambda, K(\lambda) h_\lambda \rangle}{\norm{h_\lambda}^2} = +\infty.
\end{equation}
We construct a function by defining
\begin{equation}
    g_\lambda := \tilde g_\lambda (E(\bvec{x}, \bvec{v}) , L(\bvec{x}, \bvec{v})), \quad \tilde g_\lambda:= \left.\frac {\omega(E,L)\sqrt{|\phi'(E)|}}{\sqrt{\omega(E,L)^2 - \lambda}}\right|_{\hat\Lambda},
\end{equation}
where $\hat\Lambda$ is the set of values of $(E,L)$ defined in Lemma~\ref{lemma_technical}.
Then, we take as a test function 
\begin{equation}
    h_\lambda(\bvec x, \bvec v) =  g_\lambda (\bvec{x}, \bvec{v}) \sin(\theta_r(\bvec{x}, \bvec{v})).
\end{equation}
$h_\lambda$ is odd, bounded, and we have that 
\begin{equation}
    \frac{\cL}{\sqrt{-\cL^2 - \lambda}}h_\lambda =  \frac{\omega(E,L)}{\sqrt{\omega^2(E,L) - \lambda}}g_\lambda\cos(\theta_r) 
    = \frac{g_\lambda^2}{\sqrt{|\phi'(E)|}}  \cos(\theta_r).
\end{equation}
Note that by hypothesis, $ \norm{g_\lambda}_{L^2}^2$ diverges as $\lambda \to \omega_*^2$. 
Now, we consider $\mu_\lambda$ defined in \eqref{eq:def_mu_lambda} and obtain
$$
 \mu_\lambda(h_\lambda) (r) := \int g_\lambda^2 \cos(\theta_r) \di \bvec v.
$$
This computation motivates the definition of
the functions
$$
\gamma_\lambda (r) := \mu_\lambda(h_\lambda)(r) / \norm{g_\lambda}_{L^2}^2,
$$
such that, since $ \norm{g_\lambda}_{L^2}^2 \ge  \norm{h_\lambda}_{L^2}^2$, 
\begin{equation}
    \label{eq:Rayleigh-quotient}\frac{\langle  h_\lambda, K(\lambda) h_\lambda \rangle}{ \norm{h_\lambda}_{L^2}^2} \ge \frac{\langle  h_\lambda, K(\lambda) h_\lambda \rangle}{ \norm{g_\lambda}_{L^2}^2} = \frac{\cD(\mu_\lambda(h_\lambda), \mu_\lambda(h_\lambda))}{\norm{g_\lambda}_{L^2}^2} = \norm{g_\lambda}_{L^2}^2 \cD(\gamma_\lambda, \gamma_\lambda).
\end{equation}
If $\cD(\gamma_\lambda, \gamma_\lambda)$ is unbounded along some sequence of $\lambda$'s, we have proven \eqref{eq:Rayleigh-diverges}.

Hence, we assume that $\cD(\gamma_\lambda, \gamma_\lambda)$ is bounded. We will prove that in this case, $\cD(\gamma_\lambda, \gamma_\lambda)$ is bounded uniformly away from zero, which implies the divergence of $\norm{K(\lambda)}$ in view of \eqref{eq:Rayleigh-quotient}.
Since $\cD(\cdot , \cdot )$ extends to a scalar product in $H^{-1}(B(0,R_0))$, there is a sequence $\lambda_n$ such that $\gamma_{\lambda_n}$ converges weakly to some  $\gamma_0 \in H^{-1}$.
Since norms are weakly lower semi-continuous,
$$
\cD(\gamma_0, \gamma_0) \le \liminf_n \cD(\gamma_{\lambda_n}, \gamma_{\lambda_n}).
$$
We will show below that $\gamma_0 \neq 0$. Assuming this, for the moment, we conclude that
$$
\liminf_{n \to \infty} \frac{\langle  h_{\lambda_n}, K({\lambda_n}) h_{\lambda_n} \rangle}{ \norm{h_{\lambda_n}}_{L^2}^2} \ge \cD(\gamma_0, \gamma_0)  \liminf_{n \to \infty} \norm{g_{\lambda_n}}_{L^2}^2 = +\infty.
$$

\noindent\textbf{Proof that $\gamma_0 \neq 0$.} 
To this end, take a non-negative function $\chi$ in $C^{\infty}$, such that $\chi(r)= 0$ if  $r < r_1$ and $\chi(r)=1$ if $r \in (r_2,R_0)$. 
We use the positivity of $\cos(\theta_r)$ in the support of $\chi$,and then $\cos(\theta_r)\ge 1/2$ for $r\ge r_2$, to obtain
\begin{align*}
    \langle \gamma_\lambda, \chi \rangle 
    &= 4\pi
    \norm{g_\lambda}_{L^2}^{-2 }\int_{\R_+} \chi(r) \mu_\lambda (g_\lambda)(r) r^2 \di r \\
    &\ge  4\pi \norm{g_\lambda}_{L^2}^{-2 }\int_{(r_2,R_0)} \chi(r)  \left(\int g_\lambda^2 \cos(\theta_r) \di \bvec v \right) r^2 \di r \\
    &\ge  2\pi \norm{g_\lambda}_{L^2}^{-2 }\int_{(r_2,R_0)} \left(\int  g_\lambda^2  \di \bvec v \right) r^2 \di r . 
\end{align*}

For an even function in $\theta_r$, we can integrate over $\theta_r \in (0,\pi)$, where there is a one-to-one correspondence between the variables $(r,v, w)$ and $(E,L, r)$.
We write the $\bvec v$-integral as an integral over $E$ and $L$, with Jacobian $ J(E,L, r):= [2(E-U(r))- L^2/r^2]_+^{-1/2} L/r^2 $ and apply Fubini's theorem to obtain
\begin{align*}
   \int_{r_2}^{R_0}  &\left( \int  g_\lambda^2(r, \bvec{v})  \di \bvec v \right) r^2 \di r \\
   &= 4 \pi \int_{\hat\Lambda} \tilde g_\lambda^2(E,L) \left(\int_{r_2}^{r_+(E,L)} J(E,L,r)  r^2 \di r \right) \di E\di L.
   \end{align*}
Since $\theta_r (E,L,r_2) \ge \pi/3$, 
the inner integral is comparable to the integral over all values of $r$. 
Concretely, we define $\mathsf{r_2}(E,L) \ge r_2$ such that $\theta_r(E,L, \mathsf{r_2})= \pi/3$.
For $r\le \mathsf{r_2} (E,L)$, $\cos(\theta_r)\ge \cos(\pi/3)= 1/2$. By 
using \eqref{eq:intro_frequencies}, we find
\begin{align*}
   \int_{r_2}^{r_+(E,L)} J(E,L,r)  r^2 \di r &\ge  L \int_{\mathsf{r_2}(E,L)}^{r_+(E,L)} \frac{1}{ \sqrt{2(E-U(r))- L^2/r^2}}  \di r \\
   & \ge  \frac{L T(E,L)}{6 } \\
   &= \frac{1}{3}  \int_{r_-(E,L)}^{r_+(E,L)} J(E,L,r) r^2 \di r.
\end{align*}
Hence
\begin{align*}
    \langle \gamma_\lambda, \chi \rangle 
   &\ge \norm{g_\lambda}_{L^2}^{-2 }\frac{(4\pi)^2}{6} \int_{\hat\Lambda}  \tilde g_\lambda^2(E,L) \left( \int_{r_-(E,L)}^{r_+(E,L)} J(E,L,r) r^2 \di r   \right) \di E\di L \\
   & \ge\frac{1}{6} .
\end{align*}
This allows to conclude
$$
 \langle \gamma_0, \chi \rangle =\lim_{n \to \infty } \langle \gamma_{\lambda_n}, \chi \rangle > 0. \qedhere
$$
\end{proof}

\begin{proof}[Proof of Lemma~\ref{lemma_technical}]

The set $\hat \Omega_0$  in $\R^2$ is  bounded by the lines $E= E_{0}$, $L= 0$, and a curve $L= L_{\max}(E)$, for some increasing function $L_{\max}$. For $\epsilon > 0$, $\hat \Omega_\epsilon$ stays away from the boundary $L=L_{\max}(E)$.
For each fixed pair $(E,L) \in \hat \Omega_\epsilon$, the $r$-variable takes values between the turning points $r_+(E,L)$, corresponding to $\theta_r = 0$, and $r_-(E,L)$, corresponding to $\theta_r = \pi$. 
We  define two functions $\mathsf{r}_2(E,L)> \mathsf{r}_1(E,L)$ such that $\theta_r(E,L,  \mathsf{r}_1(E,L))= \pi/2$ and $\theta_r(E,L,  \mathsf{r}_2(E,L))= \pi/3$. A sketch of this situation is given in Figure~\ref{fig:onion-orbits}.
Since $ \mathsf{r}_2$ and $\mathsf{r}_1$ are continuous functions inside $\hat{\Omega}_0$, there exists a neighborhood $V_{E,L}$ of any given $(E, L) \in \hat \Omega_0$ such that  $\max_{V_{E,L}} \mathsf{r}_1 <\min_{V_{E,L}} \mathsf{r}_2$.

This stays true on the boundary of $\hat \Omega_\epsilon$ : for pairs of the form $(E_0, L)$ with $L < L_{\max}(E_0)$, we use that the action-angle variables are well-defined for all $E<0$.
On the part of the boundary with $L= 0$, the trajectory collapses to one half of the line trajectory from $r_+$ to $0$. Still, $\lim_{L\to 0} \theta_r (E,L,r)$ is well-defined
for $r>0$ and $\mathsf{r_1}$,$\mathsf{r_2}$ have continuous extensions. 

Hence, the open sets $\{V_{E,L} | E,L \in \hat\Omega_\epsilon\}$ cover the compact set $\adh \hat\Omega_\epsilon$. After extracting a finite number of these sets, there is some $(E_*,L_*)$ such that \ref{it:divergence_point} holds on $\Lambda = \{(\bvec{x},\bvec{v}) \in \Omega_\epsilon| (E,L) \in V_{E_*,L_*}\}$.
Next, taking $r_1 = \max_{V_{E_*,L_*}} \mathsf{r}_1$, point \ref{it:suff small} is satisfied, and $r_2 = \min_{V_{E_*,L_*}} \mathsf{r}_2$ ensures point~\ref{it:bounded orbits}.
\end{proof}

\section{Calculation for polytropes}\label{sec:polytropes}

We finally prove Corollary~\ref{cor:polytropes}.
Under the hypotheses stated in the Corollary, we can bound
$$
\omega(E,L) - \omega_* \ge a(E_0-E) + b L^2 
$$
for some $a, b >0$ and all $(E,L) \in \Omega_0$. Since $\omega$ is bounded and strictly positive in $\Omega_0$, we can then estimate
$$
\rho_* (r) \le C \int_{E(x,v) < E_0} \frac{(E_0-E)^{n-1}}{(E_0-E) + c L^2} \di \bvec v := C \tilde \rho(r)
$$
for suitable constants $C$ and $c$ independent of $x$.

\begin{lemma}
For $\tilde \rho$  defined above, and for any $s \in (0,\min(n,1))$, there is a constant $C$ independent of $r \in (0,R_0)$ such that 
\begin{equation*}
     \tilde \rho(r) \le C r^{2-2s}(E_0-U(r))^{n-1/2}
     \end{equation*}
\end{lemma}
This upper bound shows that $|\bvec x|^{-1} \rho_*(\bvec x)$ is integrable for all values of $n >0$ and proves Corollary~\ref{cor:polytropes}.
\begin{proof}
    As before, we change variables from $\bvec v$ to $v, w$ and express $v, w$ in terms of $E,L$.
    We obtain 
    $$
\tilde \rho(r) = \frac{4 \pi}{r^2} \int_{U(r)}^{E_0}\int_0^{L_{\max(E)}} \frac{(E_0-E)^{n-1} }{((E_0-E) + c L^2)\sqrt{2(E- U(r)) - L^2/r^2}} L\di L \di E.
    $$
Changing variables to $t = L^2/(2 r^2(E-U(r))$ and defining $\alpha(E,r) = 2 c r^2 (E-U(r))/(E_0 -E)$ gives
    $$
\tilde \rho(r) = \frac{4 \pi}{\sqrt{2}} \int_{U(r)}^{E_0}(E_0-E)^{n-2} \sqrt{E-U(r)}\int_0^{1} \frac{\di t }{(1 + \alpha(E,r) t )\sqrt{1 -t}}  \di E.
    $$
The inner integral can be computed in closed form, so we are left with
       $$
\tilde \rho(r) = 4 \pi \sqrt{2} \int_{U(r)}^{E_0} \frac{(E_0-E)^{n-2} \sqrt{E-U(r)} }{\sqrt{\alpha(E,r)(1+\alpha(E,r))}} \tanh^{-1}\left(\sqrt{\frac{\alpha(E,r)}{\alpha(E,r) + 1}} \right)\di E.
    $$

In order to estimate the behaviour in $r$ of this integral, it is convenient to take $\alpha$ as the variable. 
We use the identities
\begin{align*}
\di E &= \frac{(E_0-E)^2}{ 2 cr^2(E_0-U(r))} \di \alpha ,\\
E_0 - E &= \frac{2cr^2}{2cr^2 + \alpha} (E_0-U(r)), \quad 
E- U(r) =  \frac{\alpha}{2cr^2 + \alpha} (E_0-U(r)).
\end{align*}
To obtain
\begin{equation} \label{eq:int_alpha}
     \tilde \rho(r) = \tfrac{ 4 \pi \sqrt{2}(E_0-U(r))^{n-1/2}}{2 c r^2} 
     \int_{0}^{\infty}\left( \tfrac{2cr^2}{2cr^2 + \alpha}\right)^n\left(\tfrac{\alpha}{2cr^2 + \alpha}\right)^{1/2}
    \tfrac{ \tanh^{-1}\left(\sqrt{\tfrac{\alpha}{\alpha + 1}}\right)}{\sqrt{\alpha(1+\alpha)}} \di \alpha.
\end{equation}
For $n < 1$, we bound the first fraction by $2cr^2/\alpha$ and the second fraction by $1$, to obtain
\begin{equation*}
     \tilde \rho(r) \le C \frac{(E_0-U(r))^{n-1/2}}{(2 c r^2)^{1-n}} 
     \int_{0}^{\infty}
    \frac{ \tanh^{-1}\left(\sqrt{\tfrac{\alpha}{\alpha + 1}}\right)}{\alpha^n\sqrt{\alpha(1+\alpha)}} \di \alpha.
\end{equation*}
 The last integral is finite: for $\alpha \in (0,1)$, we use the bound $\tanh^{-1}(t) \le 2 t$ to check that the function is integrable. For $\alpha > 1$, we can use $\tanh^{-1}\left(\sqrt{\tfrac{\alpha}{\alpha + 1}}\right) \le C \ln (\alpha) $.

 Finally, if $n\ge 1$, fix any $s \in (0,1)$. Returning to \eqref{eq:int_alpha}, bound the first parenthesis in the integrand by $(2cr^2/\alpha)^{s}$, and repeat the previous arguments.
\end{proof}  

\paragraph*{Data Availability Statement} The manuscript has no associated Data.
\paragraph*{Conflicts of Interest} The authors declare to have no conflicts of interest.

\end{document}